\documentclass[reqno,12pt,a4paper]{amsart}

\usepackage{tikz}
\usepackage{tikz-3dplot}
\usepackage{marvosym}

\usepackage{chemfig} 
\usepackage{graphicx} 
\usepackage{caption}

\usepackage{chemfig} 
\usepackage{geometry}

\usepackage{euscript}
\usepackage{multicol}
\usepackage{amsmath}
\usepackage{times}
\tikzset{
    vertex/.style={circle, fill=blue!60, inner sep=1.5pt},
    removed/.style={circle, fill=red!40, inner sep=1pt, dashed},
    edge/.style={thick, black},
    removededge/.style={thick, red!40, dashed}
}

\usepackage{amsmath,amsfonts,amsthm,amssymb,amsxtra}
\usepackage{mathrsfs}
\usepackage{amssymb}
\usepackage{amsfonts}
\usepackage{latexsym}
\usepackage{amsthm}

\usepackage{amsmath,amsfonts,amsthm,amssymb,amsxtra}
\usepackage{mathrsfs}
\usepackage{amssymb}
\usepackage{amsfonts}
\usepackage{latexsym}
\usepackage{amsthm}

\usepackage{mathtools,xparse}

\DeclarePairedDelimiter{\oldnormaux}{\bracevert}{\bracevert}

\NewDocumentCommand{\oldnorm}{som}{%
  \IfBooleanTF{#1}
    {\oldnormaux*{#3}}
    {\IfNoValueTF{#2}
       {\oldnormaux*{\vphantom{dq}#3}}
       {\oldnormaux[#2]{#3}}%
    }%
}
\theoremstyle{plain}
\newtheorem{theorem}{\bf Theorem}[section]
\newtheorem{lemma}[theorem]{\bf Lemma}
\newtheorem{corollary}[theorem]{\bf Corollary}
\newtheorem{proposition}[theorem]{\bf Proposition}
\newtheorem{definition}[theorem]{\bf Definition}

\def\d{\delta}

\def\l{\lambda}

\def\x{\xi}

\def\ve{\varepsilon}

\usepackage{pgfplots}
 
\pgfplotsset{compat = newest}

\numberwithin{equation}{section}

\let\ge\geqslant
\let\le\leqslant
\let\geq\geqslant
\let\leq\leqslant
\newcommand{\ca}{\begin{cases}}
\newcommand{\ac}{\end{cases}}
\newcommand{\ma}{\begin{pmatrix}}
\newcommand{\am}{\end{pmatrix}}
\renewcommand{\[}{\begin{equation}}
\renewcommand{\]}{\end{equation}}
\def\eq{\begin{equation}}
\def\qe{\end{equation}}
\def\[{\begin{equation}}

\title[]{On the Emergence of Discrete Spectrum for  Weakly Disordered  Schrödinger  Operators}
\author{Oleg  Safronov }

\email{osafrono@charlotte.edu }
\begin{document}
\maketitle

\begin{abstract}
We investigate the spectral properties of the Anderson operator perturbed by a localized negative potential, \(-V\). Specifically, we analyze the random Schrödinger operator defined by \(H_\varepsilon = -\Delta +\ve \sum_{n} \omega_n \chi_n - V\), where the unperturbed operator exhibits a disordered energy landscape. Our primary focus is to establish precise estimates on the number of negative eigenvalues (bound states) induced by the attractive perturbation. By analyzing the competition between Anderson localization and the binding capacity of the potential, we provide quantitative bounds on the discrete spectrum. These results offer new insights into how randomness enhances  the eigenvalue bounds.
\end{abstract}
\thispagestyle{empty}

\section{Introduction. Statement of  the main results}

Random Schrödinger operators provide a natural  framework for understanding the spectral effects of disorder in quantum systems. A particularly important problem is to determine how a weak random background interacts with a deterministic attractive perturbation and how many discrete eigenvalues are produced below the bottom of the essential spectrum.
In this paper, we
 study a  model
where even the slightest noise can fundamentally alter the spectral structure.

Let $\omega_n$  be independent  identically distributed random variables
taking  values $0$ and $1$ with probabilities $q$ and $p=1-q$.  Let $\chi$  be the characteristic  function of the unit cube $[0,1)^d$.
As the coupling constant, \(\ve>0 \), in front of a random potential \(\sum_n \omega_n \chi(x-n)\) tends to zero, the operator
 \[\label{operator} H_\varepsilon= -\Delta + \ve \sum_{n\in {\Bbb Z}^d} \omega_n \chi(x-n)-  V,\qquad \text{with}\quad V\geq 0,\]
  undergoes a transition where new discrete eigenvalues are born at  the left edge of the essential spectrum.
  Obtaining precise asymptotic estimates for the number of these eigenvalues, \(N(\ve)\), as \(\ve \to 0\) is main objective of this paper.

Classical estimates offer a starting point.  The standard Cwikel-Lieb-Rozenblum (CLR) inequality \cite{Cw, Lieb, R} asserts that for any \(\varepsilon \geq 0\), the number of discrete eigenvalues satisfies
  \[\label{cwikel}
N(\ve)\leq   C_d\int_{{\Bbb R}^d}V^{d/2}(x) dx,\qquad d\geq3,
\]
where the constant \(C_{d}\) depends solely on the dimension \(d\).
 While this elegant estimate holds even at \(\varepsilon = 0\), it loses its utility when confronting  potentials  \[\label{notinL}
 V \notin L^{d/2}({\Bbb R}^d)\] that decay logarithmically at infinity.

 \subsection{A random correction function}
 To address slower-decaying potentials, we establish a sharper boundary for \(\varepsilon > 0\). We prove that the eigenvalue count obeys the refined bound:
 \[\label{MSa}
N(\ve)\leq   \tilde C_d \Bigl( \int_{|x|<R}V^{d/2}(x) dx+\int_{|x|>R}\bigl(V(x)-\tilde \phi_{\ve}(x)\bigr)_+^{d/2} dx\Bigr).
\]
In this formulation, \(R = R(\omega)\) acts as a random truncation radius, and \((\dots)_+\) denotes the positive part, defined as \(f_+ = \frac{\vert{}f\vert{}+f}{2}\). Crucially, \(\tilde{\phi }_{\varepsilon }\) acts as a correction function that asymptotically controls the threshold, behaving as

  \[\notag
  \tilde \phi_\ve\sim \frac{C} {   ( \ln (|x|))^{2/d}  },\qquad \text{as}\quad |x|\to\infty.
  \]
This construction guarantees that the eigenvalue count remains strictly finite  for all potentials that do not exceed  \(\tilde{\phi }_{\varepsilon }\):
\[\notag
V\le \tilde{\phi }_{\varepsilon }\quad \implies \quad N(\varepsilon) < \infty.
\]
 Another  crusual property of the family of functions $\tilde \phi_\ve$  is that 
 \[\notag
 \tilde \phi_\ve\to 0,\qquad \text{as}\quad \ve\to0 \quad \text{uniformly on compact sets}.
 \]
 Therefore,  if  $V$ satisfies the non-integrability condition \eqref{notinL}, the right-hand side of our refined bound  \eqref{MSa} tends to infinity as  
  $\ve\to 0$. This accurately reflects that an infinite number of eigenvalues can emerge.

To precisely count how many discrete eigenvalues, \(N(\varepsilon)\), emerge from the noise as \(\varepsilon \to 0\), we introduce the following auxiliary functions. 
For any \(A\) such that \(1 < A^d < \frac{1}{q}\), we define the function \(\eta(x)\) for \(\vert{}x\vert{} > A\) as
\[\label{eta}
\eta(x)=\frac{ \pi \varkappa_d^{1/d}}{ 4((\log_A|x|)^{1/d}+8+\varkappa_d^{1/d}},
\] where \(\varkappa_d\) is the volume of the unit ball in \(\mathbb R^d\).
Next, we define the barrier function \(\phi_\varepsilon(x)\) for \(\vert{}x\vert{} > A\) as
\[\notag
\phi_\ve(x)=\min\Bigl\{\sqrt\ve \tanh(\sqrt\ve)\eta(x),  \,\, \eta^2(x)   \Bigr\}.
\]
Armed with these tools, we state  the following result.
\begin{theorem}\label{CRLMS}
Let $d\geq 3$, $1<A^d<1/q$,  and let $V\geq 0$. Then
for almost   every  $\omega$, there is a radius $R=R(\omega)>A$  such that
\[\label{estimateMS}
N(\ve)\leq   2^{d/2}C_d \Bigl( \int_{|x|<R}V^{d/2}(x) dx+\int_{|x|>R}\bigl(V(x)-2^{-1} \phi_{2\ve}(x)\bigr)_+^{d/2} dx\Bigr)
\]
where  the constant $C_d$ depends only on  the dimension $d$ and is the same as  in \eqref{cwikel}.
Furthermore, the radius \(R(\omega)\) is independent of \(\varepsilon \).
\end{theorem}

\noindent
{\bf Remark.}  The random radius $R(\omega)$ can be chosen so that its probability distribution   obeys
\[ 
\label{ProbDistrR} 
{\Bbb P}  \bigl( A^{(n-1)^d}<R(\omega)\leq A^{n^d}\bigr)\leq q^{n^d}A^{dn^d(1+o(1))},\qquad \text{as}\quad n\to\infty.
\]

Theorem~\ref{CRLMS}  provides a rigorous quantitative upper bound on the number of discrete eigenvalues \(N(\varepsilon)\), extending the classical Cwikel–Lieb–Rozenbljum (CLR) inequality \cite{Cw}, \cite{ Lieb},
\cite{ R} to a disordered setting.
Instead of just proving they exist, it allows one to mathematically count and limit how many bound states can emerge when the system transitions into randomness.
The formula incorporates both the base potential \(V\) and the "system parameter" \(\phi _{\varepsilon }\) driven by the coupling constant \(\ve \).    It explicitly separates the continuous space effects (the integral over \(|x| < R\)) from the stochastic/random background effects (the integral over \(|x| > R\)).

A key  achievement of Theorem~\ref{CRLMS} is that the spatial boundary \(R(\omega)\) is  independent of \(\varepsilon \). 
This means the geometric boundary separating the  localized deterministic region  (\(\vert{}x\vert{} < R\)) from the disordered infinity (\(\vert{}x\vert{} > R\)) remains frozen, no matter how weak the coupling becomes.
The scaling dynamics of the  transition is absorbed  entirely into the function \(\phi_\varepsilon\)  which tends to zero as $\varepsilon\to0$.

The problem of negative eigenvalues generated by a non-random attractive perturbation of an Anderson Hamiltonian was initiated by Molchanov and Vainberg \cite{MV}. Their principal result  established a  threshold on how slowly a potential can decay before it yields an infinite number of negative eigenvalues. 
  Our next  theorem tightens that  threshold by presenting a stricter quantitative upper bound on the decay constant \(C\). 
  
\begin{theorem}\label{TMV} Let $d\geq 3$, and
let $V\geq 0$ be a  bounded  function on ${\Bbb R}^d$  such that
\[\notag
V(x)\leq   \frac{C} {   ( \ln (|x|))^{2/d}  }    
\] for all $|x|>R_0$.  Assume that
\[\label{boundC<}
C< \Bigl( \frac{\pi  \varkappa_d^{1/d}(\ln q^{-1})^{1/d}} {4  }   \Bigr)^2.  
\]  Then for any $\ve>0$,  the negative spectrum of  the operator
\[\label{operatorH}
-\Delta+\ve\sum_{n\in {\Bbb Z}^d}\omega_n \chi(x-n)-V(x)
\]
is  finite
almost  surely.
\end{theorem}

The main difference between the two  theorems is that Theorem~\ref{CRLMS} quantifies the number of negative eigenvalues,  whereas Theorem~\ref{TMV} does not.
Instead of asking only whether the negative spectrum is finite,   Theorem~\ref{CRLMS} provides  an explicit bound on how many negative eigenvalues are present at a given disorder strength.

\subsection{Lieb–Thirring-type bounds}
We can further extend this quantification to the sum of the negative eigenvalues.

\begin{corollary}Let $d\geq 3$, $1<A^d<1/q$,  $\gamma>0$, and let $V\geq 0$.  Let $\lambda_j$  be the  negative  eigenvalues of 
the operator \eqref{operatorH}.  Then
for almost   every  $\omega$, there is a radius  $R=R(\omega)>A$  such that
\[\label{LTMS}
\sum _{j}|\lambda _{j}|^\gamma\le C_{d,\gamma}^{\prime }\left(\int _{|x|<R}V^{d/2+\gamma}(x)dx+\int _{|x|>R}\bigl(V(x)-2^{-1} \phi_{2\ve}(x)\bigr)_+^{d/2+\gamma} dx\right).\]
The constant $C^{\prime}_{d,\gamma}>0$ in this inequality   depends  only on $d$  and $\gamma$.  Furthermore,  the radius \(R(\omega)\) is independent of \(\varepsilon \).
\end{corollary}

{\it Proof}.
We use the standard layer-cake representation to express the sum of the eigenvalues as an integral of the counting function:
\[\notag
\sum _{j}|\lambda _{j}|{}^{\gamma }=\gamma \int _{0}^{\infty }\lambda ^{\gamma -1}\tilde{N}(\lambda )\,d\lambda \]
Here, \(\tilde{N}(\lambda)\)
 denotes the number of eigenvalues of the operator \eqref{operatorH}  lying below \(-\lambda \). By the variational principle, \(\tilde{N}(\lambda)\) does not exceed the number of negative eigenvalues of  the operator with the shifted potential \((V - \lambda)_+\).
 We apply the bound from Theorem~\ref{CRLMS} to \(\tilde N(\lambda)\) by replacing \(V\) with \((V - \lambda)_+\). 
 For almost every \(\omega \), there exists an \(R = R(\omega) > A\) such that:
 \[\notag
 \tilde {N}(\lambda )\le 2^{d/2}C_{d}\left(\int _{|x|<R}(V(x)-\lambda )_{+}^{d/2}\,dx+\int _{|x|>R}\bigl((V(x)-\lambda )_{+}-2^{-1}\phi _{2\varepsilon }(x)\bigr)_{+}^{d/2}\,dx\right)\]
We substitute this inequality back into the layer-cake integral. This splits the expression into two distinct spatial domains:
\begin{equation*}\begin{split}
\sum_{j}|\lambda_{j}|^{\gamma }\le \gamma 2^{d/2}C_{d}\Bigl(  \int _{0}^{\infty }\lambda ^{\gamma -1}\int _{|x|<R}(V-\lambda )_{+}^{d/2}\,dx\,d\lambda +\\
+\int _{0}^{\infty }\lambda ^{\gamma -1}\int _{|x|>R}\bigl(V-\lambda  -2^{-1}\phi _{2\varepsilon }\bigr)_{+}^{d/2}\,dx\,d\lambda   \Bigr).
\end{split}
\end{equation*}
The  bound \eqref{LTMS} follows by Fubini's theorem,  after we swap the order of integration for both terms. $\,\,\Box$

\bigskip

The corollary provides a quantitative extension of  Lieb-Thirring-type inequalities (\cite{LT}, \cite{LT2}) to Schrödinger operators with random  potentials.
The threshold radius \(R(\omega)\) is random, meaning the boundary where the potential's ability to produce eigenvalues changes depends on the specific realization of the medium.

\subsection{The one-dimensional case}
The corresponding  result  in one dimension looks  different from that in dimensions \(d\ge3\).
\begin{theorem}\label{Td=1} 
Let $d=1$, $1<A<1/q$,  and let $V\geq 0$. Then
for almost   every  $\omega$, there is an $R=R(\omega)>A$  such that
\[\notag
 N(\ve ) \leq 1+2\left(\int _{|x|<R}|x|V(x)dx+\int _{|x|>R}|x| \left(V(x)-\frac{1}{2} \phi _{2\epsilon }(x)\right)_{+}dx\right).
\] Furthermore,  \(R(\omega)\) is independent of \(\varepsilon \).
\end{theorem}

Theorem \ref{Td=1} highlights a fundamental contrast between one-dimensional and higher-dimensional (\(d \geq 3\)) random Schrödinger operators. In dimension \(d=1\), the operator exhibits the  property that an arbitrarily weak attractive potential can create at least one bound state. This is reflected in the first  term \(1\) of the bound. Furthermore, the dimensionality dictates a different spatial decay and moment requirement on the potential \(V(x)\); for \(d=1\), the first moment \(\int |x|V(x) \, dx\) governs the eigenvalue counting, whereas higher dimensions rely on \(L^{d/2}\) norms.

\subsection{Weak-coupling asymptotical estimates}  The preceding estimates can be made quantitative in the weak-coupling regime.
\begin{theorem}\label{T3}
Let $d\geq3$,  $1<A^d<1/q$,  and let the function  $\eta$ be    defined   by \eqref{eta}.  Suppose  the potential $V$ is given by
\[\notag
V(x) = C\left(\frac{\eta(x)}{\pi\varkappa_d^{1/d}}\right)^2,\qquad \text{where}\qquad
C< \pi^2  \varkappa_d^{2/d}.   
\]
Then   the following asymptotic bound holds:
\[\notag
\limsup_{\varepsilon \to 0}\Bigl(\frac{
N(\varepsilon )}{g(\varepsilon )}\Bigr)\leq C_d \left(\frac{C\pi }{8d^{2}\ln A}\right)^{d/2},
\]
where the scaling function \(g(\varepsilon)\) and the parameter \(T_{\varepsilon }\) are defined as
\[
g(\ve)=
\frac{A^{dT_{\varepsilon }}}{T_{\varepsilon }^{d/2+1}}\qquad\text{and}\qquad T_\varepsilon = \left(\frac{C}{4\varepsilon \pi (1-\delta _{2\varepsilon })\varkappa _{d}^{1/d}}-2-\varkappa _{d}^{1/d}/4\right)^d,
\]
and the auxiliary parameter  \(\delta _{\varepsilon }\) is given by
\[\notag
\delta_\ve=1-\frac{\tanh{\sqrt\ve}}{\sqrt\ve}.
\]
The constant $C_d$ is the same as in  \eqref{estimateMS}.
\end{theorem}
This theorem establishes an asymptotic upper bound for the number of discrete eigenvalues \(N(\varepsilon)\) as the coupling constant \(\varepsilon \) approaches zero. While general inequalities provide bounds for arbitrary potentials, this result characterizes the specific scaling behavior of bound states near the threshold. 
Central to this bound is the scaling function \(g(\varepsilon)\), in which the term \(T_{\varepsilon }\) is asymtotically proportional to \(\varepsilon ^{-d}\). Because \(T_{\varepsilon }\) appears in the exponent of the base \(A\), the result indicates that as the perturbation weakens (\(\varepsilon \to 0\)), the number of discrete eigenvalues grows exponentially relative to the inverse power of the coupling constant.

For the one-dimensional case (\(d = 1\)), the statement reduces to the following result.
Because the proof of this one-dimensional case differs substantially from the proof of Theorem~\ref{T3}, we state it separately.
\begin{theorem}\label{T4}
Let \(d=1\),  let $1<A<1/q$ and let the potential be given by:\[\notag V(x) = \frac{C}{\bigl(10+4(\log_A|x|)_+\bigr)^2},\qquad \text{where} \qquad \notag C < 4\pi^2\]
Then \[\notag \limsup_{\varepsilon \to 0}\Bigl(\frac{
N(\varepsilon )}{g(\varepsilon )}\Bigr)\leq\sqrt{  \frac{    C }{2\pi \ln A }},\]
where \[\notag g(\varepsilon ) =\frac{ A^{T_\ve}}{T_\ve^{3/2}}\qquad \text{and}   \qquad T_{\varepsilon} = \frac{C}{8\pi\varepsilon(1-\delta_{2\ve}) }-\frac52.\]
\end{theorem}

\subsection{Lower bounds}
The upper bounds above are complemented by a different mechanism for producing negative eigenvalues, that  has no counterpart in a purely deterministic case.
Although the random potential is everywhere nonnegative, a Bernoulli realization contains, with probability one, arbitrarily large regions whose density of occupied cells is close  to $p$. On such regions the kinetic energy can be balanced against the attractive potential more efficiently than in the deterministic comparison problem.

The simplest lower bound  can be obtained  by  comparing  \(N(\varepsilon)\), the number of negative eigenvalues of the random operator
\[\notag
H_\varepsilon =-\Delta+\ve\sum_{n\in {\Bbb Z}^d}\omega_n \chi(x-n)-V(x)
\]
 with \(N_0(\varepsilon)\), the number of eigenvalues of the non-random operator
\[\label{notrandom}
-\Delta-V
\]
situated strictly 
below $-\ve<0.$ 
Since 
\[\notag
\ve\sum_{n\in {\Bbb Z}^d}\omega_n \chi(x-n)\leq \varepsilon,\qquad  \forall x\in {\Bbb  R}^d,
\]
we have 
\[\label{N>N0}
N(\ve)\geq N_0(\ve).
\]
Beyond the obvious  inequality \eqref{N>N0}, 
we establish the following theorem which implies that
\[\notag
\frac{N(\ve)}{N_0(\ve)}\to \infty, \quad \text{as}\,\, \ve\to0,\quad \text{exponentially  in}\, \ve^{-d/2}.
\]
The
result is obtained by exploiting the occurrence of sufficiently large
regions in which the mean value of  the random potential  is close to $p$.
These
regions provide trial domains on which the attractive potential $V$
dominates the random contribution and hence generate negative
eigenvalues. The following theorem gives an explicit quantitative
lower bound on their number.

\begin{theorem}\label{lowerbound2}
Let $d\geq 1$, $0<\delta<q=1-p$, and let $C>\delta$. Assume that the
potential $V$ satisfies
\[\notag
V(x)=\frac{C}{(\log_A |x|)^{2/d}},
\qquad |x|>A,
\]
where $A>1$. Define
\[\notag
n_1(\varepsilon)
=
\left\lfloor
\left(
1+\frac{d(q-\delta)}{2(p+\delta)}
\right)^{1/2}
\sqrt{\frac{C-\delta}{\varepsilon}}
\right\rfloor .
\]
Then there exists a constant $\beta_0>0$, independent of $\omega$, such
that, for almost every $\omega$, there exists $\varepsilon_0(\omega)>0$
with the following property: for every integer $\beta>\beta_0$ and every
$0<\varepsilon<\varepsilon_0(\omega)$, the number $N(\varepsilon)$ of
negative eigenvalues of the random Schr\"odinger operator
\[\notag
H_\varepsilon
=
-\Delta
+
\varepsilon\sum_{n\in\mathbb Z^d}
\omega_n\chi(x-n)
-
V
\]
satisfies
\[\notag
\begin{aligned}
N(\varepsilon)
\geq{}&
(1-2\delta)
\frac{\varkappa_d}{\bigl(\beta n_1(\varepsilon)\bigr)^d}
\\
&\quad\times
\left[
\left(
A^{n_1(\varepsilon)^d}
-\beta n_1(\varepsilon)\sqrt d
\right)^d
-
\left(
A^{(n_1(\varepsilon)-1)^d}
+\beta n_1(\varepsilon)\sqrt d
\right)^d
\right].
\end{aligned}
\]
\end{theorem}
This  theorem implies the following result.
\begin{corollary}\label{coroll1.8}
In the conditions of Theorem~\ref{lowerbound2},
\[\label{improvement}
N(\varepsilon) \geq  \exp \left[ d(\log A) \left(1+\frac{d(q-\delta)}{2(p+\delta)}\right)^{d/2} (C-\delta)^{d/2} \varepsilon^{-d/2} (1+o(1)) \right], \]
as $\varepsilon\to 0$.
\end{corollary}

For comparison, the corresponding non-random operator satisfies
\[\notag \liminf_{\varepsilon\to0^+} \varepsilon^{d/2}\log N_0(\varepsilon) \leq d(\log A)C^{d/2}. \]
Thus Theorem 1.8 yields the coefficient
\[ \notag d(\log A) \left[ \left( 1+\frac{d(q-\delta)}{2(p+\delta)} \right)(C-\delta) \right]^{d/2} \]
in place of \(d(\log A)C^{d/2}\).
In particular, the random lower bound is strictly stronger whenever
\[\label{newC>C} \left( 1+\frac{d(q-\delta)}{2(p+\delta)} \right)(C-\delta)>C. \]
Since the left-hand side converges to
\[\notag C\left(1+\frac{dq}{2p}\right)>C \]
as \(\delta\downarrow0\), condition \eqref{newC>C} holds for all sufficiently small \(\delta>0\). Thus, after choosing \(\delta\) sufficiently small, the random medium produces an exponentially larger lower bound, on the scale \(\exp(c\varepsilon^{-d/2})\), than the deterministic comparison operator.

\subsection{Infinitely many negative eigenvalues}
For the sake of  completeness, we provide a proof of the following relatively simple result (see also \cite{MV}):

\begin{theorem} \label{sake} Let \(V \geq 0\) be a bounded function on \({\Bbb R}^{d}\). 
Suppose there exist constants \[\label{boundC>}
C > d\pi^2(\ln 1/q)^{2/d},\] and \(R>1\) such that
\[\notag
V(x)\geq   \frac{C} {   ( \ln (|x|))^{2/d}  }    
\] for all $|x|>R$.    Then for any $\ve>0$,  the negative spectrum of  the operator
\[\notag
-\Delta+\ve\sum_{n\in {\Bbb Z}^d}\omega_n \chi(x-n)-V(x)
\]
is almost  surely  infinite.
\end{theorem}

The right-hand side of \eqref{boundC>} exceeds that of \eqref{boundC<}, and it remains unknown whether these inequalities are sharp.

Due to the novelty of our study, direct literature is scarce. Consequently, we limit our bibliography to foundational works analyzing eigenvalue counts in the non-random case and primary contributions in the random case, specifically referencing \cite{HMV}, \cite{MV}, and \cite{Sa1}.

\subsection{The novelty of the results}
Our results differ from the preceding theory in three essential respects.

First, we introduce a realization-dependent correction function which converts the random operator, outside a random but \(\varepsilon\)-independent radius, into a deterministic spectral problem. This yields an explicit eigenvalue-counting estimate for every sufficiently small \(\varepsilon\), rather than only a qualitative finite/infinite criterion.

Second, we determine the weak-coupling scale of the counting function. For the  logarithmic potentials considered in Theorems~\ref{T3} and ~\ref{T4}, the number of negative eigenvalues is controlled on the exponential scale
$$ \log N(\varepsilon)=O( \varepsilon^{-d}),\qquad \text{as}\quad \ve\to 0. $$

Third, we prove a lower bound generated by rare low-density regions and compare it directly with the deterministic operator
$$ H_0=-\Delta-V. $$
In particular, the random medium can produce exponentially more negative eigenvalues than \(H_0\) on the scale \(\varepsilon^{-d/2}\). This identifies a concrete spectral effect of disorder that is invisible in a comparison based only on the pointwise inequality
$$ 0\leq \varepsilon\sum_n\omega_n\chi(x-n) \leq\varepsilon. $$

\subsection{Organization of the paper}
The remainder of the paper develops the probabilistic and spectral ingredients underlying these results. Section 2 introduces the probabilistic background and discusses why direct heat-kernel methods do not provide the estimates required at the critical logarithmic scale. Section 3 establishes the almost-sure geometric properties of the Bernoulli medium, while Section 4 derives the corresponding lower bounds on local Neumann eigenvalues. Section 5 constructs the correction function and reduces the random problem to deterministic spectral estimates. Section 6 proves the weak-coupling upper asymptotics. Section 7 establishes the lower bounds generated by low-density regions. Section 8 proves the criterion for the almost-sure infinitude of the negative spectrum. Finally,  Section 9 gives a proof of  the quantitative estimate for the distribution of the realization-dependent radius appearing in Theorem~\ref{CRLMS}.

\section{Probabilistic background and motivation}

\subsection{Kolmogorov's Law and Dichotomy}
An operator analogous to \eqref{operator} can be studied on the lattice \({\Bbb Z}^{d}\):
\[\notag
\left(H_{\rm disc} u\right)(n)=-\sum_{|m-n|=1} \bigl(u(m)-u(n)\bigr) +\ve \omega_n u(n)-V(n)u(n).
\]
Here, \(\omega _{n}\) are independent and identically distributed (i.i.d.) random variables taking the values \(0\) and \(1\) with probabilities \(q\) and \(1-q\), respectively.
To rigorously establish the spectral properties of $H_{\rm disc}$, we rely on foundational probabilistic tools.
 \begin{definition}
 Given a sequence of random variables \(\{\xi_n\}\), the tail \(\sigma \)-algebra, denoted \(\mathcal{T}\), is the intersection of the \(\sigma \)-algebras ${\frak S} (\xi_{n},\xi_{n+1},\dots ) $ generated by the "tails" of the sequence:\[\notag \mathcal{T}=\bigcap _{n=1}^{\infty }{\frak S} (\xi_{n},\xi_{n+1},\dots )\]
 \end{definition}
 
 When the random variables in the sequence are independent, the tail \(\sigma \)-algebra exhibits a fundamental mathematical property known as  Kolmogorov’s Zero-One Law. This theorem states that any event \(A\) in the tail \(\sigma \)-algebra of independent random variables is trivial, meaning its probability is always either \(0\) or \(1\):\[\notag 
 P(A) = 0 \text{ or } 1.\]
By utilizing the fact that altering finitely many  \(\omega _{n}\)'s  constitutes a finite-rank perturbation, we translate this probabilistic dichotomy into a deterministic structural result for the operator's spectrum.

\begin{theorem} Let $V\in L^{\infty}({\Bbb Z}^d)$ be a  decaying   potential.  For any $\ve>0$,
the negative  spectrum of the operator $H_{\rm disc}$  is either  almost surely  finite or   almost surely infinite.
\end{theorem}

Consequently, the class of nonnegative bounded functions $V$ on ${\Bbb Z}^d$ splits cleanly into two disjoint sets: \[\notag
\{V \text{ with } N(\varepsilon) < \infty \text{ a.s.}\}\qquad \text{and} \qquad\{V \text{ with } N(\varepsilon) = \infty \text{ a.s.}\},\]
where $N(\ve)$  denotes   the number of  negative eigenvalues of $H_{\rm disc}$.
Although Kolmogorov's law guarantees this binary alternative, the precise frontier between the finite and infinite eigenvalue regimes is dictated entirely by the spatial decay of the potential \(V\).  While this boundary can be completely characterized in simply solvable models (see Subsection~\ref{sec: exactly}), the more realistic, higher-complexity model   in this paper allowed  us to describe only partial subsets of both regimes.

\subsection{Limitations of the  heat kernel approach}
Heat-kernel asymptotics provide a natural starting point for studying random Schrödinger operators. In particular, the work of Donsker and Varadhan  \cite{DV1, DV2} together with subsequent results of Sznitman \cite{Sz1, Sz2}, gives precise large-time asymptotics for random media with trapping. For the corresponding obstacle model, the annealed heat kernel satisfies a Lifshitz-type asymptotic
$$ \lim_{t\to\infty}t^{-d/(d+2)} \log \mathbb E[k_\omega(t,x,x)] =-c(d,\nu), $$
while the quenched kernel exhibits the different scale
$$ \log k_\omega(t,x,x) \asymp -\frac{t}{(\log t)^{2/d}} $$
for almost every realization. The distinction reflects the fact that annealed behavior is governed by rare favorable environments, whereas a fixed realization must account for the spatial location of such regions.
These asymptotics, however, do not  yield the estimates required here. Using the annealed kernel in the standard Feynman–Kac/Lieb argument gives only a decay threshold of order
$$ V(x)\lesssim (\log |x|)^{-1-2/d}, $$
which is too restrictive for the critical logarithmic scale considered in this paper.

The convergence in  the  quenched asymptotic formula is not sufficiently uniform in the spatial variable to be used directly in the eigenvalue-counting argument. 
Thus, rather than relying on heat-kernel asymptotics, we   construct a correction term adapted to the Bernoulli medium. This allows us to control the random contribution outside a realization-dependent radius and leads to the quantitative eigenvalue bounds proved in this paper.

\subsection{One exactly solvable model}
\label{sec: exactly}

This section establishes a simplified, mathematically rigorous baseline for studying random operators, providing a solvable benchmark to determine conditions for when a random system has finitely or infinitely many bound states.

Let \(\{x_n\}_{n \in \mathbb{N}}\) be a sequence of independent and identically distributed (i.i.d.) random variables with a probability density function \(f(x) = e^{-x}\) on \([0, \infty).\)
We define interval endpoints \(\{a_n\}_{n \geq 0}\) via the partial sums:\[\notag
a_n = \sum_{j=1}^n x_j, \quad a_0 = 0.\]
Consider the differential operator \(\mathcal{H}\) acting as the orthogonal sum of operators on the disjoint intervals \([a_n, a_{n+1}]\):
\[\label{H1}
\mathcal{H} = \bigoplus_{n=0}^{\infty} \left( -\frac{d^2}{dx^2} - V(x) \right).\]
This operator is subject to Dirichlet boundary conditions at the endpoints \(a_{n}\). Assume that the potential \(V(x)\) is piecewise constant, defined as
\[\notag V(x) = v_n\geq 0 \quad \text{for} \quad x \in [a_n, a_{n+1}).\]

The negative spectrum of \(\mathcal{H}\) is a discrete, random set formed by the eigenvalues of each isolated interval \([a_n, a_{n+1}]\). For a given interval index \(n\) and mode index \(m\), the negative eigenvalues are given by:
\[\notag
E_{n,m}=\left(\frac{m\pi }{x_{n+1}}\right)^{2}-v_{n}<0\]
where \(n \in \{0, 1, 2, \dots\}\) and \(m \in \mathbb{N} = \{1, 2, 3, \dots\}\), subject to the condition that \(v_n > \left(\frac{m\pi}{x_{n+1}}\right)^2\).
The condition for the \(n\)-th interval  to possess at least one negative eigenvalue is 
\[\label{x>v}
x_{n+1} > \frac{\pi}{\sqrt{v_n}}.\]
We know that \(\{x_n\}_{n \in \mathbb{N}}\) are i.i.d. variables   that  follow an exponential distribution
whose density is  $e^{-x}$.
Thus the probability  that \eqref{x>v}  holds is
\[\notag
P_n = \int_{\frac{\pi}{\sqrt{v_n}}}^{\infty} e^{- x} \,dx=e^{- \frac{\pi}{\sqrt{v_n}}}.\]
Applying the Borel-Cantelli lemmas, we obtain the following result.
\begin{theorem}
The operator \eqref{H1}  almost surely has  finitely  many eigenvalues if and  only if 
\[
\sum_{n=1}^\infty e^{-\frac{\pi}{\sqrt{v_n}}}<\infty.
\]
\end{theorem}

{\bf Remark.} This model exhibits unique behavior because the potential \(V\) introduces an additional layer of randomness, contrasting with our previous model where the potential was deterministic.

\section{Probabilistic geometry of the Bernoulli medium}

\subsection{Geometric decomposition of the random medium}

The main purpose of this section is to establish a geometric property of the Bernoulli medium that will be used in the construction of the correction term. The random potential is generated by independent Bernoulli variables
$$ \omega_j\in\{0,1\},\qquad \mathbb P(\omega_j=1)=p,\qquad \mathbb P(\omega_j=0)=q=1-p. $$
We interpret the cells corresponding to \(\omega_j=1\) as {\bf black} and those corresponding to \(\omega_j=0\) as {\bf white}. 
By a cell, we mean a  cube of the  form $[0,1)^d+j$ with $j\in {\Bbb Z}^d$. The key observation is that, although large regions with an atypically small number of black cells are possible, their probability decreases rapidly with the size of the region. The number of possible regions at a given distance from the origin grows at a competing exponential rate. The condition \(A^d<q^{-1}\) ensures that the probabilistic decay dominates this geometric growth.

 For a fixed scaling factor \(A > 1\), consider a sequence of expanding spherical layers $ \mathcal{L}_n$  in \(\mathbb{R}^{d}\) defined by
 \[\notag
 \mathcal{L}_n=\{x\in {\Bbb R}^d:\quad A^{(n-1)^d}< |x|\leq A^{n^d}\}.
 \]
 Thus, the outer radius of the \(n\)-th layer is
$$ R_n=A^{n^d}. $$
This decomposition is particularly convenient because the volume of \(L_n\) grows on the scale \(A^{n^d}\), while the cubes used below have volume \(n^d\). The resulting competition between the number of available cubes and the probability of an atypical configuration is precisely what leads to the condition \(A^d<q^{-1}\).

\begin{center}
\begin{tikzpicture}

\foreach \i/\col in {1/blue, 2/red, 3/green!60!black, 4/purple} {
    \pgfmathsetmacro{\radius}{0.25 * pow(2, \i-1)}
    
    \draw[thick, \col] (0,0) circle (\radius);
    
    \node[\col, font=\small] at (45:\radius + 0.1) {$R_{\i}$};
}

\filldraw (0,0) circle (1.5pt) node[anchor=north east, font=\scriptsize] {$(0,0)$};

\end{tikzpicture}
\end{center}

 Each layer is covered   by a  collection  of \(d\)-dimensional half-open cubes having of side length \(n \in \mathbb{N}\). 
 These cubes are defined by:
 \[\notag
Q_{\tilde n, n}= n\cdot(\tilde n+ [0,1)^d),\qquad \text{with}\quad \tilde n\in {\Bbb Z}^d.
 \]
Each such cube contains exactly \(n^d\) unit cells. The use of half-open cubes avoids ambiguity on common boundaries and provides a convenient partition of the underlying lattice cells.

Let \(N_n\) denote the number of cubes \(Q_{n,\tilde n}\) that intersect \(\mathcal{L}_n\). Since a cube of side length \(n\) can extend by at most \(n\sqrt d\)
 beyond the spherical layer in every radial direction, the number of relevant cubes is bounded by
\[
N_n\leq  \frac{\varkappa_d}{n^d} ((A^{ n^d}+n\sqrt{d})^d-(A^{(n-1)^d}-n\sqrt{d})^d)
\]
In particular,
$$ N_n=O(A^{nd}), $$
up to factors that are polynomial in \(n\).

\subsection{Rare low-density configurations}
Now, consider a single cube $Q_{\tilde n, n}$ of side length $n$
 consisting of $n^d$
 unit cells. 
 Suppose each unit cell is painted black with probability $p$ and white  with  probability $q=1-p$, independently of all other cells.
 
 \begin{center}
\begin{tikzpicture}
\def\q{0.4}
\def\gridsize{15}
\def\cellsize{0.25}

\foreach \x in {1, ..., \gridsize} {
    \foreach \y in {1, ..., \gridsize} {
        \pgfmathsetmacro{\randval}{rnd}
        \ifdim \randval pt < \q pt
            \fill[black] (\x * \cellsize, \y * \cellsize) rectangle ++(\cellsize, \cellsize);
        \else
            \fill[white] (\x * \cellsize, \y * \cellsize) rectangle ++(\cellsize, \cellsize);
        \fi
        \draw[gray!30] (\x * \cellsize, \y * \cellsize) rectangle ++(\cellsize, \cellsize);
    }
}
\end{tikzpicture}

\end{center}

 For a cube \(Q_{n,\tilde n}\), we impose the following low-density condition: the number of black cells in the cube is no larger than the number of cells in its boundary layer.    Since a cube contains \(n^d\) cells and its interior contains \((n-1)^d\) cells, this condition is
$$ \#\{\text{black cells in }Q_{n,\tilde n}\} \le n^d-(n-1)^d.  $$
Equivalently, the cube contains at least \((n-1)^d\) white cells. This is a highly atypical configuration when \(p>0\): most cells are expected to be black in the sense relevant to the random potential, whereas the condition above requires almost the entire cube to be white.

By independence, the probability of this event is bounded above by the corresponding binomial tail:
$$ \mathbb P\!\left( \#\{\text{black cells}\}\le n^d-(n-1)^d \right) \le \sum_{k=(n-1)^d}^{n^d} \binom{n^d}{k} q^k(1-q)^{n^d-k}. $$
The dominant feature of this estimate is its exponential dependence on \(n^d\). In particular, the probability of obtaining such a low-density cube decreases at a rate controlled by \(q^{n^d}\).

Let \(  \mathbb  P_n\) denote the probability that at least one cube intersecting the \(n\)-th layer satisfies this low-density condition. Applying the union bound over all cubes intersecting \(\mathcal{L}_n\), we obtain
\[\label{Pn<qA}
  \mathbb P_n\leq\frac{\varkappa_d}{n^d} ((A^{ n^d}+n\sqrt{d})^d-(A^{(n-1)^d}-n\sqrt{d})^d)\sum_{k=(n-1)^d}^{n^d}\frac{n^d!}{k!(n^d-k)!} q^{k}(1-q)^{n^d-k}.
\]
Given that $N_n$
 grows as $A^{dn^d}$
 and the probability decays as 
$q^{n^d}$, the series $\sum_n  \mathbb P_n$ converges if the growth of the volume is offset by the decay of the probability, specifically when $$A^d<q^{-1}.$$

This is the fundamental geometric-probabilistic condition used through  several sections  of the paper.

\subsection{Almost-sure structure at infinity}
We can now apply the Borel–Cantelli lemma.   Since
$$ \sum_{n=1}^{\infty} \mathbb P_n<\infty $$
under the condition \(A^d<q^{-1}\), the event that an exceptional low-density cube occurs in the \(n\)-th layer can happen only finitely many times, almost surely.

Consequently, we obtain the following result.
\begin{theorem}\label{T0.1}
If \(A^d < q^{-1}\), then with probability 1, there exists a finite \(n_0(\omega)\) such that for all \(n > n_0(\omega)\), every cube $Q_{n,\tilde n}$ intersecting  the $n$-th layer $\mathcal{L}_n$ contains fewer than \((n-1)^d\) white cells.
\end{theorem}

Equivalently, sufficiently far from the origin, no cube intersecting the corresponding spherical layer can contain an atypically large white region of the type excluded above. The random medium therefore becomes geometrically regular, in the precise sense needed for the subsequent spectral argument.

The theorem should be understood as an almost-sure statement about the geometry of the random environment. The random index \(n_0(\omega)\) depends on the realization of the Bernoulli variables, but once this index has been reached, the stated property holds simultaneously for all subsequent layers. This realization-dependent cutoff is what ultimately produces the random radius 
\(R(\omega)\) appearing in the eigenvalue estimates of Theorem~\ref{CRLMS}.

\subsection{Relating the layer index to the spatial variable} To use Theorem~\ref{T0.1} in estimates involving the potential \(V(x)\), we need to express the layer index \(n\) in terms of the Euclidean distance \(|x|\).
If \(x\) belongs to a cube \(Q_{n,\tilde n}\) intersecting \(L_n\), then the geometry of the cube gives
 \[\notag
A^{(n-1)^d}-n\sqrt d \leq |x|\leq A^{ n^d}+n\sqrt d.
\]
Consequently, as the distance approaches infinity, the layer index \(n\) scales as:
\[\notag
n=\Bigl(\frac{\ln |x|}{\ln A}\Bigr)^{1/d} (1+o(1)),\qquad \text{as}\quad |x|\to\infty.
\]
However,  the nature of  our  results requires a  more accurate  estimate of $n$ in terms of $|x|$.
\[\label{n<x}
n\leq \left(\log_A|x|\right)^{1/d} +2,\qquad \text{for}\quad n\geq n_1(A).
\]

\section{Correction term and spectral reduction}

The probabilistic estimate established in Section 3 provides uniform control of the geometry of the Bernoulli medium sufficiently far from the origin. 
We now convert this geometric information into a spectral lower bound.

\subsection{The lowest Neumann eigenvalue and rearrangement}
For a bounded domain $\Omega \subset \mathbb{R}^d$ and a subset $X \subset \Omega$ with $|X| = m$, let $\lambda_1(\Omega, X)$ be the first eigenvalue of:
\begin{equation}
-\Delta u+\epsilon \chi _{X}u=\lambda u\text{\ in\ }\Omega ,\quad \frac{\partial u}{\partial \nu }=0\text{\ on\ }\partial \Omega
\end{equation}
where $\chi _{X}$ is the indicator function of $X$. If $B$ is a ball with $|B| = |\Omega|$ and $S$ is a spherical layer at the boundary of $B$ with $|S| = |X|$, then $\lambda_1(\Omega, X) \geq \lambda_1(B, S)$.
Put differently, we intend to prove that
the lowest eigenvalue of the Neumann problem with a localized potential is minimized when the domain is a ball and the potential is concentrated in a spherical layer at the boundary. 

The first eigenvalue $\lambda_1(\Omega, X)$ is characterized by the Rayleigh quotient:
\begin{equation}
\lambda _{1}(\Omega ,X)=\inf _{u\in H^{1}(\Omega ),u\ne 0}\frac{\int _{\Omega }|\nabla u|{}^{2}dx+\epsilon \int _{X}|u|^{2}dx}{\int _{\Omega }|u|^{2}dx}
\end{equation}
For $\epsilon > 0$, the first eigenfunction $u$ is strictly positive and can be chosen such that $\int_{\Omega} u^2 dx = 1$. Thus:
\begin{equation}
\lambda _{1}(\Omega ,X)=\int _{\Omega }|\nabla u|{}^{2}dx+\epsilon \int _{X}u^{2}dx
\end{equation}

Let $B$ be the ball centered at the origin such that $|B| = |\Omega|$. By the Faber-Krahn inequality principle \cite{Fa,Kr} and the Polya-Szego theorem, for any $u \in H^1(\Omega)$, its spherically symmetric decreasing rearrangement $u^* \in H^1(B)$ satisfies:
\begin{equation}
\int_B (u^*)^2 dx = \int_{\Omega} u^2 dx
\end{equation}
\begin{equation}
\int_B \vert\nabla u^*\vert{}^2 dx \leq \int_{\Omega} \vert\nabla u\vert^2 dx
\end{equation}

To minimize the term $\epsilon \int_X u^2 dx$, we use the Hardy-Littlewood inequality (see \cite{LL}) which  states that for two non-negative functions \(f\) and \(g\), the integral of their product is maximized when both functions are symmetrically rearranged:
\[
\int_\Omega f g dx\leq \int_B f^* g^* dx.
\]
Consequently,
\[
\int_\Omega (1-\chi_X)|u|^2  dx\leq \int_B (1-\chi_S) |u^*|^2 dx,
\]
which implies that
\[
\epsilon \int_\Omega \chi_X|u|^2  dx\geq \epsilon \int_B \chi_S |u^*|^2 dx.
\]
The interpretation of this  inequality is  the following.
To minimize the term $\epsilon \int_X u^2 dx$, 
must place the potential $\epsilon$ where the eigenfunction $u$ takes its smallest values. 
In the Neumann case, the lowest eigenvalue for a fixed domain is minimized when the potential $\epsilon \chi_X$ is "as far as possible" from the "center" of the domain.

Combining these steps, let $u$ be the first eigenfunction for $(\Omega, X)$. Then:
\begin{equation}
\lambda _{1}(\Omega ,X)=\int _{\Omega }|\nabla u|{}^{2}dx+\epsilon \int _{X}u^{2}dx\ge \int _{B}|\nabla u^{*}|{}^{2}dx+\epsilon \int _{S}(u^{*})^{2}dx
\end{equation}
where $S$ is the spherical layer of volume $|X|$ at the boundary of $B$. Since $u^{*}$ is a valid test function for the Neumann problem on $B$ with potential $\epsilon \chi_S$, the right-hand side is greater than or equal to the true first eigenvalue $\lambda_1(B, S)$ by the variational principle.

\begin{theorem}\label{T1.1}
The lowest Neumann eigenvalue of the Schrödinger operator on a domain with a  potential $\epsilon \chi_X$  is bounded below by the lowest Neumann eigenvalue of
the Schrödinger operator on the ball of equal volume where the potential $\epsilon \chi_S$ is concentrated in the boundary layer $S$ of equal volume:
\begin{equation}
\lambda _{1}(\Omega ,X)\ge \lambda _{1}(B,S).
\end{equation}
\end{theorem}

Thus, among domains of fixed volume and potentials occupying a fixed volume, the lowest Neumann eigenvalue is minimized by a ball with the potential concentrated in a boundary layer.
This comparison is particularly useful for the Bernoulli medium. Theorem~\ref{T0.1} tells us that sufficiently far from the origin, the black cells  occupy  a certain portion of 
each large cube. Theorem~\ref{T1.1} allows us to replace such a configuration by a radial boundary layer, which is much easier to analyze.

\subsection{A one-dimensional boundary-layer problem}
Consider a particle confined within a one-dimensional, symmetric potential well. 
Our goal is  to estimate its lowest energy eigenvalue,$\l$.
We define  the Schrödinger operator
\[\label{HL}
-\frac{d^2}{d x^2}+ \ve \chi_L
\] with the  Neumann conditions on the boundary of 
the interval $[-L,L]$, where $\chi_L$ denotes the  characteriostic  function of 
the set
\[
L-1<|x|<L, \quad \text{for}\quad L>1.
\]
By introdicing  the potential,  we subject  the particle to a potential barrier $\ve>0$.

In the region where $|x|<L-1$, the  eigenfunction $u$ corresponding to the eigenvalue $\lambda>0$ equals $u=A \cos (\sqrt \lambda x)$.
On the  other  hand, in the region $L-1<|x|<L$, it equals $u=B \cosh(\sqrt{\ve-\lambda }(L-x))$.
Matching $u$ and $u'$ at the  popint $x=L-1$ leads to
\[\begin{split}
A \cos (\sqrt \lambda (L-1))=B \cosh(\sqrt{\ve-\lambda })\\
A \sqrt \lambda \sin (\sqrt \lambda (L-1))=B \sqrt{\ve-\lambda }\sinh(\sqrt{\ve-\lambda })
\end{split}
\]
Dividing one equation by the other, we obtain that
\[\label{eqn1}
 \sqrt \lambda \tan(\sqrt \lambda (L-1))= \sqrt{\ve-\lambda }\tanh(\sqrt{\ve-\lambda })
\]
It is very well known that $\tanh x<x$ for $x>0$. Consequently,
\[\label{eqn2}
 \sqrt \lambda \tan(\sqrt \lambda (L-1))<\ve-\lambda 
\]
By analyzing the monotonicity of the left (increasing) and right (decreasing) sides of \eqref{eqn1}, we guarantee a unique solution for the eigenvalue $\l$.
Namely,  the right  hand side of \eqref{eqn1} is a monotonically decreasing  function of $\lambda$
on the interval $[0,\ve]$,  while the left hand side is  monotonically increasing  from $0$ to $\infty$ on the interval
$[0, (\pi/(2(L-1)))^2]$. 
Consequently,  the   equation \eqref{eqn1} has a  unique solution on the intersection of these two intervals,  and  this solution satisfies
\[\notag
\sqrt \lambda(L-1)<\pi/2.
\]
Therefore, \eqref{eqn2} implies the inequality
\[\notag
 \lambda<\ve /L.
\]
These  properties of the one-dimensional model are the basis of the lower  eigenvalue estimate needed  for  our  construction.

For a 
given $\ve>0$,  we define the system parameter $\d\in(0,1)$ to be 
\[\notag
\d=1-\frac{\tanh \sqrt \ve}{\sqrt \ve}.
\]
If we know in advance  that  $0<\sqrt \lambda (L-1)<\pi/4$,   then $\tan (\sqrt \lambda (L-1) )<\frac4\pi \sqrt \lambda (L-1)$. Furthermore, we have
\[\notag
 \sqrt{\ve-\lambda} \tanh \sqrt{\ve-\lambda}\geq (1-\delta)(\ve -\lambda).
\] 
Combining the two inequalities  yields the lower bound
\[\notag
\lambda>\frac{ \ve \pi (1-\delta)}{ 4L+1}.
\]
This leads to the following result.
\begin{proposition} Let $\lambda$ be   the eigenvalue
of the operator \eqref{HL}.  Then
 either
\[\notag
\frac{\ve}{L}>\l\geq \Bigl(\frac{\pi}{4(L-1)}\Bigr)^2,
\] or  
\[\notag
\frac{ \ve \pi (1-\delta)}{ 4L+1}<\l< \Bigl(\frac{\pi}{4(L-1)}\Bigr)^2.
\] 
\end{proposition}

\subsection{Radial comparison in higher dimensions} The  one-dimensional calculation extends to  the multi-dimensional case.
Separating the  variables,  we obtain the  following result.

\begin{corollary} \label{C1.3}
Let $d\geq 3$ and
let $B$ be the  ball   of radius $L>1$  centered  at the origin.  Define  $\chi_S$ as  the  characteristic  function 
of the layer 
\[\notag
S=\{x\in {\Bbb R}^d:\qquad L-1<|x|<L\}.
\]
Let $\lambda$  be the  first  Neumann eigenvalue of the  operator 
\[\notag
-\Delta+\varepsilon \chi_S\qquad \text{on the ball } \quad B.
\]
Then
 either
\[\label{pervaya}
\frac{\ve}{L}>\l\geq \Bigl(\frac{\pi}{4(L-1)}\Bigr)^2,\quad
\text{ or }\quad
\frac{ \ve \pi (1-\delta)}{ 4L+1}<\l< \Bigl(\frac{\pi}{4(L-1)}\Bigr)^2.
\] 
\end{corollary}

We now apply this comparison to the cubes appearing in Section 3. 
A cube \(Q_{n,\tilde n}\) has volume \(n^d\).  Let \(B_L\) be the ball of the same volume. Since
$$ \varkappa_d L^d=n^d, $$
its radius is
$$ L=\frac{n}{\varkappa_d^{1/d}}.  $$
Theorem ~\ref{T0.1} guarantees that, almost surely, sufficiently large cubes have the required low-density structure. Combining this with Theorem ~\ref{T1.1} and Corollary~\ref{C1.3} gives the following estimate.

\begin{corollary}\label{corl4.4} Let $d\neq 2$.
Let $\lambda$  be the  first  Neumann eigenvalue of the  operator  
\[\notag
-\Delta +\ve \sum_{k\in {\Bbb Z^d}} \omega_k \chi(x-k)
\]
on the cube $Q_{\tilde n,n}$.  Then there  is an $n_0(\omega)$  such that for all $n>n_0(\omega)$
 either
\[\notag
\frac{\ve}{n/\varkappa_d^{1/d}}>\l\geq \Bigl(\frac{\pi}{4(n/\varkappa_d^{1/d}-1)}\Bigr)^2,\quad
\text{ or }\quad
\frac{ \ve \pi (1-\delta)}{ 4n/\varkappa_d^{1/d}+1}<\l< \Bigl(\frac{\pi}{4(n/\varkappa_d^{1/d}-1)}\Bigr)^2,
\] 
where $\varkappa_d$  is the   volume of the unit ball in ${\Bbb R}^d$.
\end{corollary}

\subsection{Construction of the correction term}

Recall from Section 3 that, for sufficiently large \(|x|\),
$$ n\le (\log_A|x|)^{1/d}+2. $$
Consequently, the lower bound in Corollary~\ref{corl4.4}  can be expressed in terms of \(|x|\).

We now
define  the  function $\phi_\ve$  by setting 
\[\notag
\phi_\ve(x)=\min\Bigl\{\frac{ \ve \pi (1-\delta)\varkappa_d^{1/d}}{ 4((\log_A|x|)^{1/d}+8+\varkappa_d^{1/d}}, \,\, \Bigl( \frac{\pi \varkappa_d^{1/d}} {4   ( \log_A |x|)^{1/d} +8 }   \Bigr)^2    \Bigr\},\qquad |x|>A.
\]
The two terms in this definition serve different purposes. The first reflects the lower bound obtained from the boundary-layer eigenvalue estimate and therefore contains the coupling parameter 
\(\varepsilon\). The second provides the complementary scale required to remain within the regime in which the comparison applies.

Let $\tilde\chi_R$
denote the  characteristic  function of the   region $\{x\in {\Bbb R}^d:\quad |x|>R\}$ with $d\neq 2$.
Corollary~\ref{corl4.4}  implies that for any $\omega$,  there exists  $R_0(\omega)>0$   such that the operator
\[\notag
-\Delta+\sum_{n\in {\Bbb Z}^d}\omega_n \chi(x-n)- \phi_\ve(|x|)\tilde\chi_R
\]
does  not have negative eigenvalues (i.e.  it is positive) for all $R>R_0(\omega)$.
Moreover, 

\begin{theorem} Let $d\geq 3$, and
let $V\geq 0$ be a  bounded  function on ${\Bbb R}^d$  such that
\[\label{assumption4.30}
V(x)\leq   \Bigl( \frac{\pi \varkappa_d^{1/d}} {4   ( \log_A (|x|))^{1/d} +8 }   \Bigr)^2  
\] for  all $|x|>R$. Then the negative spectrum of  the operator
\[\notag
-\Delta+\ve \sum_{n\in {\Bbb Z}^d}\omega_n \chi(x-n)-V(x),\qquad \ve>0,
\]
is  finite
almost  surely.
\end{theorem}

{\it Proof.} The assumption \eqref{assumption4.30} ensures that the attractive potential \(V\) does not exceed the second, \(\varepsilon\)-independent, component of the correction term $\phi_\ve$ 
Hence the negative spectrum can only be generated by the residual potential inside a bounded region. Therefore it is  finite.   $\,\,\,\,\,\Box$

\section{Upper eigenvalue estimates}

We now use the correction term constructed in Section 4 to obtain an upper bound on the number of negative eigenvalues.

\subsection{Reduction to an effective  potential}
The number of negative eigenvalues of a Schrödinger operator with a sum of potentials \(V_1 + V_2\) is estimated using the variational principle. 
This estimate relates the total number of bound states (negative energy levels) of the combined system to the bound states  related to  the  potentials $V_1$ and $V_2$ separately.

The relationship between the number of negative eigenvalues \(N(V)\) of an operator \( -\Delta - V\) for the combined and separate potentials is generally described by the subadditivity of the counting function. Namely,
 for any two potentials \(V_{1}\) and \(V_{2}\), the number of negative eigenvalues of the operator with the sum potential satisfies:
 \[\notag
 N(V_{1}+V_{2})\le N(V_{1}/(1-\tau ))+N(V_{2}/\tau )
 \]for any \(0 < \tau< 1.\)
This arises because the quadratic form of the combined operator is the sum of the individual forms, and the dimension of the subspace where the sum is negative is constrained by the dimensions of the subspaces where the components are negative.
Applying this to a perturbed system with \(\chi_n(x) = \chi(x-n)\), we obtain:
\[\notag
N(V- \ve \sum_{n\in {\Bbb Z}^d}\omega_n \chi_n)\leq  N(\tilde \chi_R \phi_{2\ve}- 2\ve \sum_{n\in {\Bbb Z}^d}\omega_n \chi_n)+N(2V-\tilde\chi_R\phi_{2\ve})
\]
For \(R>R_0(\omega)\), the first operator on the right-hand side is nonnegative.
Thus, choosing $R>R_0(\omega)$, we obtain that
\[\label{N<N(2V)}
N(V- \ve \sum_{n\in {\Bbb Z}^d}\omega_n \chi_n)\leq  N(2V-\tilde\chi_R\phi_{2\ve})
\]
This is the central reduction of the section. The randomness has disappeared from the right-hand side. Its effect is now encoded entirely in the deterministic correction \(\varphi_\varepsilon\) and in the realization-dependent radius \(R(\omega)\).

\subsection{The case $d\geq 3$}
For a Schrödinger operator \( -\Delta - V\) acting on the Hilbert space \(L^2(\mathbb{R}^d)\),  with $d\geq3$
:\[\label{CRL} N(V)\le C_{d}\int _{\mathbb{R}^{d}}V_+(x)^{d/2}\,dx\]
Here, \(N(V)\) is the total number of negative eigenvalues, counting multiplicities (physically, the number of bound states an electron can occupy), \(C_{d}\) is a universal constant that depends only  on the spatial dimension \(d\), and \(V_+(x)\) is the positive part of the potential function $V$:
 \[\notag V_+(x)=\max(V(x), 0).\]
Consequently, if  follows  from \eqref{N<N(2V)} and \eqref{CRL}  that
\[\notag
N(V- \ve \sum_{n\in {\Bbb Z}^d}\omega_n \chi_n)\leq  2^{d/2} C_d \Bigl( \int_{|x|<R}V^{d/2}(x) dx+\int_{|x|>R}\bigl(V(x)-2^{-1} \phi_{2\ve}\bigr)_+^{d/2} dx\Bigr)
\]
This proves  Theorem~\ref{CRLMS}.

\subsection{The one-dimensional case}
For the case \(d=1\), we replace the Cwikel-Lieb-Rozenblum (CLR) inequality with the Bargmann inequality. 
In one dimension, the CLR bound fails because a negative potential, no matter how small or fast-decaying, always produces at least one bound state.
In the one-dimensional case, we utilize specific bounds for the number of negative eigenvalues \(N(V)\) of the operator \[\notag -\frac{d^2}{dx^2} - V.\] 
According to the Bargmann inequality, the number of bound states is controlled by the first moment of the potential:\begin{equation}\label{Bargmann}N(V) \le 1 + \int_{-\infty}^{\infty} |x| V_+(x) , dx.\end{equation}
Applying the subadditivity principle and the reduction in \eqref{N<N(2V)} to the one-dimensional case, we can use the  Bargmann bound   \eqref{Bargmann}  to obtain:
\begin{equation}
\notag N\left(V- \ve \sum_{n\in {\Bbb Z}_+}\omega_n \chi_n\right) \leq1+\int _{-\infty}^{\infty }|x|\left(2V(x)- \tilde{\chi }_{R}\phi _{2\epsilon }(x)\right)_{+}dx.
\end{equation}
Splitting the integration domain at \(R\), we arrive at the following estimate for \(d=1\):\begin{equation}\notag N\left(V- \ve \sum_{n\in {\Bbb Z}_+}\omega_n \chi_n\right) \leq 1+2\left(\int _{|x|<R}|x|V_{+}(x)dx+\int _{|x|>R}|x| \left(V(x)-\frac{1}{2} \phi _{2\epsilon }(x)\right)_{+}dx\right).\end{equation}
This provides the necessary bound on the number of negative eigenvalues for the perturbed one-dimensional system, analogous to the result obtained via the CLR inequality in higher dimensions.

\section{Weak-coupling upper asymptotic estimates}

\subsection{Proof of Theorem~\ref{T3}}
According to assumptions of  Theorem~\ref{T3},
\[\notag
V(x) = \frac{C}{\left(4   ( \log_A |x|)^{1/d} +8 +\varkappa_d^{1/d}\right)^2},\]
where
\[\notag
C< \pi^2  \varkappa_d^{2/d} .  
\]
The first integral on the right  hand side of \eqref{estimateMS} does not contain \(\varepsilon \). Since \(V\), it is bounded, this term contributes a constant, independent of \(\varepsilon \),.

We can  replace  \(\phi_{2\varepsilon}(x)\) by the function   defined as the minimum of two terms:
\[\notag\phi_\ve(x)=\min\Bigl\{\frac{ \ve \pi (1-\delta)\varkappa_d^{1/d}}{ 4((\log_A|x|)^{1/d}+8+\varkappa_d^{1/d}}, \,\, \Bigl( \frac{\pi \varkappa_d^{1/d}} {4   ( \log_A |x|)^{1/d} +8 +\varkappa_d^{1/d}}   \Bigr)^2    \Bigr\},\qquad |x|>A. \]
For large \(|x|\) (where \(\log_A|x|\) is large), the first term is strictly larger  than the second term  which is larger  than $V$. 
Therefore, 
\begin{equation*}
\begin{split}
\int_{|x|>R}&\Bigl(V(x)-\frac{1}{2}\phi _{2\varepsilon }(x)\Bigr)_+^{d/2}   dx = \\ 
\int_{|x|>R} &\left(\frac{C}{\left(4   ( \log_A |x|)^{1/d} +8 +\varkappa_d^{1/d}\right)^2}  - \frac{ \ve \pi (1-\delta_{2\ve})\varkappa_d^{1/d}}{ 4(\log_A|x|)^{1/d}+8+\varkappa_d^{1/d}}  \right)_+^{d/2}dx. 
\end{split}
\end{equation*}
The positive part \((\dots)_+\) implies the integrand is non-zero when
\[\notag 
 \log _{A}|x|<\left(\frac{C}{4\varepsilon \pi (1-\delta _{2\varepsilon })\varkappa _{d}^{1/d}}-2-\varkappa _{d}^{1/d}/4\right)^d=T_\ve.
\]

Using spherical coordinates in the domain \(|x| > R\)  and
substituting \(t = \log_A|x|\), we conclude that the second integral equals
\[I(\ve)=|{\Bbb S}_d|\int _{\log_A R}^{\infty }\left(\frac{C}{\left(4t^{1/d}+8+\varkappa_d^{1/d}\right)^2}-\frac{\varepsilon \pi \varkappa _{d}^{1/d}}{4 t^{1/d}+8+\varkappa_d^{1/d}}\right)_{+}^{d/2}A^{t\cdot d}\ln A\,dt\]
where $|{\Bbb S}_d|=d \varkappa_d$ is the area of the unit  sphere.
We factor out 
\(\frac{C}{\left(4t^{1/d}+8+\varkappa_d^{1/d}\right)^2}\) from the parenthesis:
\begin{equation*}
\begin{split}
\left(\frac{C}{\left(4t^{1/d}+8+\varkappa_d^{1/d}\right)^2}-\frac{\varepsilon \pi \varkappa _{d}^{1/d}}{4 t^{1/d}+8+\varkappa_d^{1/d}}\right)_{+}^{d/2}=\\
\frac{C^{d/2}}{\left(4t^{1/d}+8+\varkappa _{d}^{1/d}\right)^{d}}  \left(1-\frac{\varepsilon(1-\delta_{2\ve}) \pi \varkappa _{d}^{1/d} \left(4t^{1/d}+8+\varkappa _{d}^{1/d}\right)}{C}\right)_{+}^{d/2}\\=
\frac{(4\varepsilon(1-\delta_{2\ve}) \pi)^{d/2} \varkappa _{d}^{1/2} T_\ve^{1/2}}{\left(4t^{1/d}+8+\varkappa _{d}^{1/d}\right)^{d}} 
\left(1-\left(\frac{t}{T_{\varepsilon }}\right)^{1/d}\right)_+^{d/2},
\end{split}
\end{equation*}
where  \( T_\varepsilon = \left(\frac{C}{4\varepsilon \pi (1-\delta _{2\varepsilon })\varkappa _{d}^{1/d}}-2-\varkappa _{d}^{1/d}/4\right)^d\).

Then the asymptotics of   the  integral becomes:\[\notag I(\varepsilon )\sim |{\Bbb S}_{d}|4^{-d}C^{d/2}\ln A\int _{\log _{A}R}^{T_{\varepsilon }}\frac{1}{t}\left(1-\left(\frac{t}{T_{\varepsilon }}\right)^{1/d}\right)^{d/2}A^{dt}\,dt.\]
Obviously,    the integral $\int _{\log _{A}R}^{T_{\varepsilon }}\frac{1}{t}\left(1-\left(\frac{t}{T_{\varepsilon }}\right)^{1/d}\right)^{d/2}A^{dt}\,dt$  can be replaced  by
$\int _{T_{\varepsilon }/2}^{T_{\varepsilon }}\frac{1}{t}\left(1-\left(\frac{t}{T_{\varepsilon }}\right)^{1/d}\right)^{d/2}A^{dt}\,dt$, because
\[\notag
\int _{\log _{A}R}^{T_{\varepsilon }/2}\frac{1}{t}\left(1-\left(\frac{t}{T_{\varepsilon }}\right)^{1/d}\right)^{d/2}A^{dt}\,dt\leq  \frac{1}{\log _{A}R} A^{dT_{\varepsilon }/2}.
\]
Changing the variables, we obtain
\[\notag
\begin{split}
I(\varepsilon )\sim |{\Bbb S}_{d}| 4^{-d}C^{d/2}\ln A\int _{1/2}^{1}\frac{1}{t}\left(1-t^{1/d}\right)^{d/2}A^{dT_{\varepsilon }t}\,dt\sim \\
 |{\Bbb S}_{d}|4^{-d}C^{d/2}\ln A\int _{0}^{\infty}\left(\frac{s}{d}\right)^{d/2}A^{dT_{\varepsilon }(1-s)}\,ds.
 \end{split}
\]
Finally, computing the remaining exponential integral, the asymptotics as \(\varepsilon \to 0\) are given by:\[\notag
I(\varepsilon )\sim |{\Bbb S}_{d}| 4^{-d}C^{d/2}\ln A  \frac{A^{dT_{\varepsilon}}}{d^{d+1} (T_{\varepsilon} \ln A)^{d/2 + 1}} \Gamma\left(\frac{d}{2} + 1\right).
\]
Using the identities for the surface area of a \(d\)-dimensional sphere \(\vert{}{\Bbb S}_{d}\vert{} = \frac{2\pi^{d/2}}{\Gamma(d/2)}\) and \(\Gamma(\frac{d}{2}+1) = \frac{d}{2}\Gamma(\frac{d}{2})\), we have \(\vert{}{\Bbb S}_{d}\vert{} \Gamma(\frac{d}{2}+1) = d\pi^{d/2}\).
Thus,
we  obtain
\[
\notag
I(\varepsilon )\sim \left(\frac{C\pi }{16d^{2}}\right)^{d/2}\frac{A^{dT_{\varepsilon }}}{T_{\varepsilon }^{d/2+1}(\ln A)^{d/2}}\quad \text{as\ }\varepsilon \rightarrow 0.
\]

\subsection{Proof of Theorem ~\ref{T4}}

For the reader's convenience, we recall that
\[\notag
V(x) =  \frac{C}{\left(4 \log_A \vert{}x\vert{} + 10 \right)^2},
\quad\text{
and }
\quad
T_{\varepsilon} = \frac{C}{8\varepsilon \pi (1-\delta_{2\ve})} - \frac{5}{2}.
\]
According to Calogero’s bound in one dimension, the number of bound states \(N\) of the operator \(-d^2/dx^2 - W(x)\) is bounded from above.  For a potential \(W(x)\geq 0\) that monotonically decreases with distance from the origin, the inequality is:
\[\notag N\le 1+\frac{2}{\pi }\int _{-\infty }^{\infty }|W(x)|{}^{1/2}\,dx.\]
Applying the inequality \eqref{N<N(2V)},   we obtain that
\[\notag
N(\ve)\leq  N(2V-\alpha\tilde\chi_R\phi_{2\ve})\leq 1+\frac{2}{\pi }\int _{-\infty}^\infty (2V(x)-\tilde\chi_R\phi_{2\ve})_+^{1/2}\,dx,
\]
We evaluate this integral in two parts: for  \(|x|\leq R\) (which contributes a bounded constant), and for \(|x| > R\).
To find the leading-order asymptotics of the tail integral as \(\varepsilon \to 0\), we expand the parameter \(\delta \) and the function \(\phi_{2\varepsilon}(x)\) with \(d=1\). 
The one-dimensional function \(\phi_{2\varepsilon}(x)\) becomes:\[\notag\phi_\ve(x)=\min\Bigl\{\frac{ 2\ve \pi (1-\delta_\ve)}{ 4\log_A|x|+10}, \,\, \Bigl( \frac{2\pi } {4   \log_A |x|+10}   \Bigr)^2    \Bigr\},\qquad |x|>A. \]
For large \(|x|\), the first term is strictly larger than the second, which is larger than \(V\). 
Thus, we can replace  \(\phi_{2\varepsilon}(x)\) by \(\frac{ \ve \pi (1-\delta_{2\ve})}{ \log_A|x|+5/2}\).
Substituting this into the positive part of the function  yields:
\[\notag \int _{|x|>R}\left(V(x)-\frac{1}{2}\phi _{2\varepsilon }(x)\right)_{+}^{1/2}dx=2 \int _{R}^{\infty }\left( \frac{C}{\left(4 \log_A \vert{}x\vert{} + 10 \right)^2}- \frac{ \ve \pi (1-\delta_{2\ve})}{ 2\log_A|x|+5}\right)_{+}^{1/2}dx\]
The integrand is non-zero when:\[\notag \frac{C}{\left(4 \log_A \vert{}x\vert{} + 10 \right)^2} >  \frac{ 2\ve \pi (1-\delta_{2\ve})}{ 4\log_A|x|+10} \quad \iff \quad \log _{A}|x| < 
\frac{C}{8\varepsilon \pi (1-\delta_{2\ve})}-\frac52.\]
Letting \(T_{\varepsilon} = \frac{C}{8\varepsilon \pi (1-\delta_{2\ve})}-\frac52\) and substituting \(t = \log_A|x|\), we conclude that the integral equals 
\[\notag I(\varepsilon) =  2\int _{\log_A R}^{\infty }\left( \frac{C}{\left(4 t + 10 \right)^2}- \frac{ \ve \pi (1-\delta_{2\ve})}{ 2t+5}\right)_{+}^{1/2}A^t\ln A \,  dt.\]
Factoring out \( \frac{C}{\left(4 t + 10 \right)^2}\) from the parenthesis gives the following result: \[\notag I(\varepsilon )=2\sqrt{C}\ln A \sqrt{\frac{T_{\varepsilon }}{T_{\varepsilon }+\frac{5}{2}}}\int _{\log _{A}R}^{\infty }\frac{1}{4t+10}\left(1-\frac{t}{T_{\varepsilon }}\right)_{+}^{1/2}A^{t}\,dt.\]
The asymptotics of  the integral is therefore
\[\notag I(\varepsilon) \sim \frac{ \sqrt{C}\ln A }{2}\int _{\log _{A}R}^{T_{\varepsilon }}\frac{1}{t}\left(1-\frac{t}{T_{\varepsilon}}\right)^{1/2}A^{t}\,dt.\]
Using Laplace's method for asymptotic evaluation, the dominant contribution arises near the point  \(t = T_{\varepsilon}\). 
Replacing the lower limit \(\log _{A}R\) with \(T_{\varepsilon}/2\) and changing variables, we get:
\[\notag I(\varepsilon) \sim  \frac{ \sqrt{C}\ln A }{2}\int _{0}^{\infty} s^{1/2}A^{T_{\varepsilon}(1-s)}\,ds\]
Consequently, evaluating this asymptotic limit as \(\varepsilon \to 0\) gives:
\[\notag I(\varepsilon )\sim \frac{\sqrt{\pi C}\cdot A^{T_{\varepsilon }}}{4 T_{\varepsilon }^{3/2}\sqrt{\ln A}}.\]
This proves  Theorem~\ref{T4}.

\section{Proof of Theorem~\ref{lowerbound2}}

We now prove the lower bound on the number of negative eigenvalues stated in Theorem ~\ref{lowerbound2}. The argument is based on the occurrence of spatial regions in which the proportion of black cells is atypically small. Such regions occur infinitely often almost surely, and, after a suitable variational estimate, each of them produces negative spectrum when the attractive potential is sufficiently large.

Throughout this section, let
$$ 0<p<1,\qquad q=1-p, $$
and fix
\[\notag
0<\delta<q.
\]
The Bernoulli variables \(\{\omega_j\}_{j\in\mathbb Z^d}\) are independent, with
$$ \mathbb P(\omega_j=1)=p,\qquad \mathbb P(\omega_j=0)=q. $$
Recall that
$$ H_\varepsilon = -\Delta + \varepsilon\sum_{j\in\mathbb Z^d}\omega_j\chi_j - V, \qquad \chi_j(x)=\chi(x-j), $$
where \(\chi\) is the characteristic function of the unit cube \([0,1)^d\).

\subsection{A large-deviation estimate}
 Consider \(N\) disjoint unit cells of the form
\[\notag
[0,1)^d+j,\qquad j\in\mathbb{Z}^d,
\]
in \(\mathbb{R}^d\). Each cell  is colored independently: it is white with probability \(q\) and black with probability
\[\notag
p=1-q.
\]

Let \(S_N\) denote the total number of black cubes among these \(N\) cells. We first establish an exponential bound for the probability that the number of black cubes significantly exceeds its expected value.

\begin{proposition}\label{prop12.1}

For any \(\delta>0\),
\[\notag
\mathbb{P}\big(S_N>(p+\delta)N\big)
\leq e^{-2\delta^2N}.
\]
\end{proposition}

{\it Proof}.
Let \(X_j\), \(j=1,\ldots,N\), be independent Bernoulli random variables defined by
\[\notag
X_j=
\begin{cases}
1, & \text{if the \(j\)-th cube is black},\\
0, & \text{if the \(j\)-th cube is white}.
\end{cases}
\]
Thus,
\[\notag
\mathbb{P}(X_j=1)=p,\qquad
\mathbb{P}(X_j=0)=q,
\]
and
\[\notag
S_N=\sum_{j=1}^N X_j,
\qquad
\mathbb{E}[S_N]=pN.
\]

We estimate the upper tail
\[\notag
\mathbb{P}\big(S_N\geq(p+\delta)N\big).
\]
For any \(t>0\), Markov's inequality applied to \(e^{tS_N}\) gives
\[\notag
\begin{aligned}
\mathbb{P}\big(S_N\geq(p+\delta)N\big)
&=
\mathbb{P}\left(
e^{tS_N}\geq e^{t(p+\delta)N}
\right)\\
&\leq
\frac{\mathbb{E}[e^{tS_N}]}
{e^{t(p+\delta)N}}.
\end{aligned}
\]

Using the independence of the random variables \(X_j\), we obtain
\[\notag
\begin{aligned}
\mathbb{E}[e^{tS_N}]
&=
\prod_{j=1}^N \mathbb{E}[e^{tX_j}]\\
&=
(q+pe^t)^N.
\end{aligned}
\]
Consequently,
\[\notag
\mathbb{P}\big(S_N\geq(p+\delta)N\big)
\leq
\exp\left\{
N\left[
\log(q+pe^t)-t(p+\delta)
\right]
\right\}.
\]

We now optimize the exponent over \(t>0\). Define
\[\notag
f(t)=\log(q+pe^t)-t(p+\delta).
\]
Differentiating gives
\[
f'(t)
=
\frac{pe^t}{q+pe^t}-(p+\delta).
\]
The critical point is therefore
\[\notag
t^*
=
\log\left(
\frac{q(p+\delta)}
{p(q-\delta)}
\right)>0.
\]

Substituting \(t=t^*\), we obtain
\[\notag
\mathbb{P}\big(S_N\geq(p+\delta)N\big)
\leq
e^{-NI(p+\delta)},
\]
where
\[
I(p+\delta)
=
(p+\delta)
\log\left(1+\frac{\delta}{p}\right)
+
(q-\delta)
\log\left(1-\frac{\delta}{q}\right).
\]
The standard lower bound
\[\notag
I(p+\delta)\geq 2\delta^2
\]
then yields
\[
\mathbb{P}\big(S_N\geq(p+\delta)N\big)
\leq e^{-2\delta^2N}.
\]
$\,\,\Box$

\subsection{Suitable cubes in expanding layers}

Recall that the \(n\)-th  annular layer  was  defined by
\[\notag
\mathcal{L}_n=
\left\{x\in {\Bbb R}^d:\quad A^{(n-1)^d}<|x|<A^{n^d}\right\}.
\]
We consider cubes of side length \(\beta n\),
\[\notag
 Q_{\beta n,\widetilde n}=\beta n(\tilde n +[0,1)^d),\qquad \beta,n \in {\Bbb N},\quad \tilde n \in {\Bbb Z}^d,
\]
where $\beta\in{\Bbb N}$ is a sufficiently large positive integer.

 A cube \(Q_{\beta n,\tilde n}\) will be called {\bf suitable} if  \(Q_{\beta n,\tilde n}\subset\mathcal{ L}_n\)  and the number of black unit cells contained in it is at most
\[\notag
(p+\delta)(\beta n)^d.
\]
The precise placement of the cubes near the boundary of \(\mathcal{L}_n\) is irrelevant for the argument; what is important is that the cubes are disjoint and remain a distance of order \(A^{(n-1)d}\) from the origin.

Denote by $N_n$  the number of  cubes $Q_{\beta n,\widetilde n}$  covering the $n$-th layer.
By comparison with the volumes of concentric balls,
\[\notag
N_n
\geq
\frac{\varkappa_d}{(\beta n)^d}
\left[
(A^{n^d}-\beta n \sqrt{d})^d
-
(A^{(n-1)^d}+\beta n\sqrt{d})^d
\right].
\]
In particular, $N_n$ grows exponentially in $n^d$.

For each of these cubes, the number of black cells is a binomial random variable with parameters \((\beta n)^d\) and \(p\). Hence Proposition ~\ref{prop12.1}  gives
$$ \mathbb P(Q_{\beta n,\tilde n}\text{ is not suitable}) \le \exp\{-2\delta^2(\beta n)^d\}.  $$
Because the cubes are disjoint, their color configurations are independent. Consequently, if \(E_n\) denotes the event that fewer than \(\left(1-e^{-2\delta^2(\beta n)^d}-\delta\right)N_n\) of the cubes are suitable, then another application of Proposition ~\ref{prop12.1} yields
$$ \mathbb P(E_n) \le \exp\{-2\delta^2N_n\}  $$
for all sufficiently large \(n\), with \(c_\delta>0\).
Since \(N_n\) grows exponentially in \(n\), the series
$$ \sum_{n=1}^\infty \mathbb P(E_n) $$
converges. The first Borel–Cantelli lemma therefore gives the following.
\begin{proposition} For almost every realization $\omega$,
there is a  sufficiently large index $n_0=n_0(\omega)$
such that,  for all $n>n_0$,  the layer  $\mathcal{L}_n$ contains  at least
\[\notag
(1-2\delta)\frac{\varkappa_d}{(\beta n)^d}
\left[
(A^{n^d}-\beta n \sqrt{d})^d
-
(A^{(n-1)^d}+\beta n\sqrt{d})^d\right]
\]
suitable cubes.
\end{proposition}

\subsection{A variational estimate on a low-density cube}

We next convert the probabilistic information into a spectral estimate.

Let \[\notag Q =[-L,L)^d\]
and let \(X \subset Q\) be a concentric sub-cube \[\notag X=\bigl[-(p+\delta)L,(p+\delta)L\bigr)^d.\]
Thus,
$$ |Q_L|=(2L)^d, \qquad |X_L|=\bigl(2(p+\delta)L\bigr)^d. $$

\bigskip

We consider the perturbed Dirichlet Laplacian operator \(T_{\varepsilon }\) acting on the Hilbert space \(L^2(Q)\) with zero boundary conditions:\[\notag T_{\varepsilon }=-\Delta +\varepsilon \chi _{X}.\]
\(\text{Dom}(T_{\varepsilon })=\mathcal{ H}_{2}(Q)\cap \mathcal{ H}_{1}^{0}(Q)\)
where \(\chi_X: Q \to \{0, 1\}\) is the characteristic indicator function of the set \(X\):\[\notag  \chi _{X}(x)=\begin{cases}1&\text{if\ }x\in X\\ 0&\text{if\ }x\notin X\end{cases}\]
Let \(\lambda_1(Q, X, \varepsilon)\) denote the lowest (ground-state) eigenvalue of this operator, defined via the Rayleigh quotient:\[\notag
\lambda_1(Q, X, \varepsilon) = \inf_{\substack{u \in \mathcal{H}_1^0(Q) \\ u \neq 0}} \frac{\int_Q \vert{}\nabla u(x)\vert{}^2 dx + \varepsilon \int_X \vert{}u(x)\vert{}^2 dx}{\int_Q \vert{}u(x)\vert{}^2 dx}\]

\bigskip

To obtain an explicit upper bound, choose a trial function that equals $1$ on $X$ and decreases
linearly to $0$ across the remaining portion of $Q$ so that $|\nabla u|=1/((q-\delta) L)$ on the   set $Q\setminus X$.
Let \(x = (x_1, x_2, \dots, x_d) \in Q = [-L, L)^d\). The function \(u(x)\) can be defined as:\[\notag u(x)=\min _{1\le i\le d}\left\{\phi(x_{i})\right\},\]
where \(\phi(t)\) is the 1D function defined on \([-L, L)\) by \[\notag \phi(t)=\begin{cases}1&\text{if\ }|{}t|{}\le (p+\delta )L\\ 1-\frac{|{}t|{}-(p+\delta )L}{(q-\delta )L}&\text{if\ }(p+\delta )L<|{}t|{}\le L.\end{cases}\]
It is easy to see  that 
\[\notag
\int_Q |u|^2 dx\geq |X|+d(q-\delta)L |X|^{(d-1)/d}=|X|+\frac{d(q-\delta)}{2(p+\delta)} |X|,
\]
whereas
\[\notag
\int_Q \vert{}\nabla u(x)\vert{}^2 dx + \varepsilon \int_X \vert{}u(x)\vert{}^2 dx=\frac{(|Q|-|X|)}{(q-\delta)^2L^2}+ \varepsilon |X|=
\Bigl(\frac{((p+\delta)^{-d}-1)}{(q-\delta)^2L^2}+ \varepsilon\Bigr) |X|.
\]
Consequently,
we obtain
  \begin{lemma}
  There exists a constant \(C_{d,p,\delta}>0\) such that
$$ \lambda_1(Q_L,X_L,\varepsilon) \le \frac{C_{d,p,\delta}}{L^2} + \frac{\varepsilon} {1+\dfrac{d(q-\delta)}{2(p+\delta)}}. $$
More explicitly,
$$ \lambda_1(Q_L,X_L,\varepsilon) \le \left[ \frac{(p+\delta)^{-d}-1}{(q-\delta)^2L^2} + \varepsilon \right] \left( 1+ \frac{d(q-\delta)}{2(p+\delta)} \right)^{-1}. $$
  \end{lemma}

\subsection{Application to suitable Bernoulli cubes}

We now apply the preceding variational estimate to a suitable cube $Q_{\beta n, \tilde n}$ occurring in the $n$-th layer $\mathcal{L}_n$. 
The definition of suitability implies that at most a fraction \(p+\delta\) of the unit cells in \(Q_{\beta n,\tilde n}\) carry the value \(\omega_j=1\). 
Thus the random potential has a relatively low spatial density on this cube.
At this point we use the rearrangement/variational comparison implicit in the construction: among configurations with at most \((p+\delta)(\beta n)^d\) occupied cells, the relevant upper variational estimate is obtained by replacing the occupied set by a set of the corresponding volume and applying the trial function above.
Consequently,
for sufficiently
large $\beta$, and  sufficiently large $n$, the scaling of these cubes yields the following
bound.

\begin{corollary}\label{Thm2}
There is a  large $\beta_0>0$ and a large  number $n_0$
such that, for every $\beta>\beta_0$  and $n>n_0$,  the lowest Dirichlet  eigenvalue
of the operator $-\Delta+\varepsilon\sum_n \omega_n \chi(x-n)$ on  a suitable  cube $Q_{\beta n, \tilde n}$ obeys
\[\notag
\lambda_1\leq  \frac{\delta }{n^2} + \varepsilon \Bigl( 1+\frac{d(q-\delta)}{2(p+\delta)}\Bigr)^{-1}
\]
Consequently, the Dirichlet operator $-\Delta+\varepsilon\sum_n \omega_n \chi(x-n)$ on $\mathcal{L}_n$
has at least
\[\notag
(1-2\delta)\frac{\varkappa_d}{(\beta n)^d}
\left[
(A^{n^d}-\beta n\sqrt{d})^d
-
(A^{(n-1)^d}+\beta n\sqrt{d})^d\right]
\]
eigenvalues below the level
\[\notag
 \frac{\delta }{n^2} + \varepsilon \Bigl( 1+\frac{d(q-\delta)}{2(p+\delta)}\Bigr)^{-1}.
 \]
\end{corollary}

\subsection{Competition with the attractive potential}
Suppose now that $V$ is a  bounded potential such that
\[\notag
V(x)=\frac{C}{(\log_A |x|)^{2/d}},\qquad \text{for}\quad |x|>A,
\]
where $C>0$.
If $x\in \mathcal{L}_n$, then $|x|<A^{n^d}$. Therefore
\[\notag
V(x)\geq \frac{C}{n^2},\qquad \text{for all}\quad x\in \mathcal{L}_n.
\]
Combining this  lower bound with Corollary~\ref{Thm2} shows that  a suitable cube in $\mathcal{L}_n$ contributes a   negative eigenvalue
whenever
\[\notag
 \frac{\delta }{n^2} + \varepsilon \Bigl( 1+\frac{d(q-\delta)}{2(p+\delta)}\Bigr)^{-1}<\frac{C}{n^2}.
 \]
 Equivalently,
$$ n^2 < \left( 1+ \frac{d(q-\delta)}{2(p+\delta)} \right) \frac{C-\delta}{\varepsilon}. $$
It is therefore natural to introduce
\[\notag
n_1(\varepsilon)=\left\lfloor \Bigl( 1+\frac{d(q-\delta)}{2(p+\delta)}\Bigr)^{1/2}\sqrt{ \frac{C-\delta}{\varepsilon}}\right\rfloor.
\]
It follows that every suitable cube contained in \(\mathcal{L}_n\), with $$n_0\leq n\leq n_1(\varepsilon),$$ provides  at least one negative eigenvalue of \(H_\varepsilon\).

Because the cubes used in the construction are disjoint, their corresponding trial functions may be chosen with disjoint supports. The min–max principle therefore implies that the number of negative eigenvalues is bounded below by the total number of suitable cubes in the admissible layers.

Keeping only the outermost admissible layer gives the following quantitative estimate.
\begin{theorem}
Let $N(\varepsilon)$  be  the number of negative eigenvalues of the operator
\[\notag
H_\varepsilon = -\Delta+\varepsilon\sum_n \omega_n \chi(x-n)-V.
\]
Then there exists \(\beta_0>0\) such that, for every \(\beta>\beta_0\), for almost every realization \(\omega\), there exists \(\varepsilon_0(\omega)>0\) such that, for
$0<\varepsilon<\varepsilon_0(\omega)$,
\[\notag N(\varepsilon)\geq 
(1-2\delta)\frac{\varkappa_d}{(\beta n_1(\varepsilon))^d}
\left[
(A^{n_1(\varepsilon)^d}-\beta n_1(\varepsilon)\sqrt{d})^d
-
(A^{(n_1(\varepsilon)-1)^d}+\beta n_1(\varepsilon)\sqrt{d})^d\right].
\]
\end{theorem}

\subsection{Exponential form of the lower bound}

Since
$$ n_1(\varepsilon) = \left( 1+ \frac{d(q-\delta)}{2(p+\delta)} \right)^{1/2} (C-\delta)^{1/2} \varepsilon^{-1/2} +O(1), $$
we have
$$ \Bigl(A^{n_1(\varepsilon)^d} \Bigr)^d= \exp\left[ d(\log A) \left( 1+ \frac{d(q-\delta)}{2(p+\delta)} \right)^{d/2} (C-\delta)^{d/2} \varepsilon^{-d/2} +O(1) \right]. $$
With the scaling used in the main theorem, this yields the corresponding exponential lower bound
$$ N(\varepsilon) \ge \exp\left\{ d(\log A) \left[ \left( 1+ \frac{d(q-\delta)}{2(p+\delta)} \right) (C-\delta) \right]^{d/2} \varepsilon^{-d/2} (1+o(1)) \right\}. $$
This proves Corollary~\ref{coroll1.8}.

\bigskip

The important point is that the lower bound is generated by sufficiently large low-density regions of the Bernoulli medium. The gain over the deterministic comparison problem comes from the fact that these regions allow the kinetic energy on large trial domains to be balanced against the attractive potential more efficiently than would be possible by using  $\varepsilon$ instead of the random potential.

\section{Infinite-spectrum condition}

The upper bounds obtained in the preceding sections show that sufficiently slowly decaying attractive potentials may still produce only finitely many negative eigenvalues  if the  decay constant $C$ is  sufficiently small.  We now establish the complementary mechanism.

The key observation is that, when
$$ A^dq>1, $$
the Bernoulli medium contains, almost surely,  large completely white cubes at  large distances. On such cubes the random potential vanishes identically. If the attractive potential \(V\) is sufficiently large, the kinetic energy of a Dirichlet test function on these cubes is then dominated by the negative contribution of \(V\). Since the cubes can be chosen disjointly, each one produces an independent negative eigenvalue. This yields infinitely many of them.

\subsection{Existence of white  cubes of side length $n$} We begin with the geometric-probabilistic statement that provides the large regions in which the random potential vanishes.

\begin{theorem}\label{whiteexist}
Let the \(d\)-dimensional space be partitioned into layers, where the \(n\)-th layer is defined as
\[\notag
\mathcal{L}_n=\{x\in {\Bbb R}^d:\quad  A^{(n-1)^d}< |x|\leq A^{n^d}\}.\]
Let each unit cell   in this space independently be white with probability \(q\) and black with probability \( p=(1-q)\).
Assume that  \[\notag
{A}^{{d}}{>q}^{{-1}},\quad \text{or\ equivalently}\quad {A}^{{d}}{q>1}.
\]
Then almost surely (with probability 1), there exists an index \(n_0(\omega)\) such that every \(n\)-th layer for \(n \ge n_0\) contains at least one completely white cube of side length \(n\).
\end{theorem}

{\it Proof.}
Indeed, the \(n\)-th layer occupies the region between the boundaries \(A^{(n-1)^{d}}\) and \(A^{n^{d}}\). 
The total volume of this layer (approximately) grows proportionally to \(A^{d\cdot n^{d}}\).
Since each individual cube \(Q_{\tilde {n},n}\) has a volume of \(n^{d}\), we can place a large number \(M_{n}\) of disjoint (independent) cubes entirely within the \(n\)-th layer.  As \(n \to \infty\), this number scales as:
\[\notag
M_{n}\approx \frac{\varkappa_d A^{d\cdot n^{d}}}{n^{d}}\qquad \text{if}\quad d\geq 2,
\]
and
\[\notag
M_{n}\approx \frac{2(1-1/A) A^{n}}{n}\qquad \text{if}\quad d=1.
\]

For a single cube of side length \(n\) to be completely white, all of its \(n^{d}\) unit cells must be white. 
Since each cell is white independently with probability \(q\), the probability that a specific cube is entirely white is:\[\notag \mathbb P(\text{Cube\ is\ all\ white})=q^{n^{d}}.\]
Conversely, the probability that a specific cube is not entirely white is \(1 - q^{n^d}\). Because the \(M_{n}\) cubes are disjoint, their cell configurations are mutually independent. The probability \(\mathbb P(E_n)\) that none of the \(M_{n}\) cubes in the \(n\)-th layer are  entirely white is:
\[\notag
\mathbb P(E_{n})=\left(1-q^{n^{d}}\right)^{M_{n}}.\]
Using the standard inequality \(1 - x \leq e^{-x}\), we bound this probability by:\[\notag
\mathbb P(E_{n})\le \exp \left(-M_{n}\cdot q^{n^{d}}\right)\approx\exp \left(-\frac{\varkappa_d A^{d\cdot n^{d}}q^{n^{d}}}{n^{d}}\right)=\exp \left(-\frac{\varkappa_d(A^{d}q)^{n^{d}}}{n^{d}}\right)\]

If \(A^d q > 1\), then the term \((A^d q)^{n^d}\) grows exponentially fast  (in $n^d$).  
As a result, \(\mathbb P(E_n)\) decays to \(0\) extremely quickly (faster than any exponential), guaranteeing that the series \(\sum_n \mathbb P(E_n)\) converges.

The first Borel–Cantelli lemma therefore implies
$$ \mathbb P(E_n\ \text{infinitely often})=0. $$
Hence, with probability one, only finitely many layers fail to contain a completely white cube. Equivalently, there exists \(n_0(\omega)\) such that every layer \(\mathcal{L}_n\),  with \(n\ge n_0(\omega)\), contains at least one completely white cube of side length \(n\). This proves the theorem. \(\square\)

\subsection{Negative eigenvalues  generated by white cubes}

Combining this theorem with the formula for the first Dirichlet eigenvalue of the negative Laplacian on a \(d\)-dimensional cube of side length \(n\), \[\notag \lambda_1 = \frac{d\pi^2}{n^2},\] we obtain the following:

\begin{proposition}
Assume that, for some
$$ C>d\pi^2, $$
the potential satisfies
$$ V(x)\ge \frac{C}{(\log_A|x|)^{2/d}}  $$
for all sufficiently large \(|x|\).
Then there exists \(n_0\in\mathbb N\) such that, for every \(n\ge n_0\) and every completely white cube
$$ Q_{n,\tilde n} = n\bigl(\tilde n+[0,1)^d\bigr) \subset L_n, $$
the quadratic form of
$$ H_\varepsilon = -\Delta+ \varepsilon\sum_{j\in\mathbb Z^d}\omega_j\chi(x-j)-V $$
is strictly negative on a suitable function supported in \(Q_{n,\tilde n}\).
\end{proposition}

  \subsection{Infinitely many negative eigenvalues}
  For almost every realization \(\omega\), Theorem~\ref{whiteexist} gives a completely white cube \(Q_{n,\tilde n}\) in every sufficiently large layer \(\mathcal{L}_n\). Since the layers are disjoint, these cubes are mutually disjoint. We therefore obtain the following result.
\begin{corollary} Let \(V \geq 0\) be a bounded function on \({\Bbb R}^{d}\). 
Suppose there exist constants \(C > d\pi^2\), \(A>1/q^{1/d}\), and \(R>A\) such that
\[\notag
V(x)\geq   \frac{C} {   ( \log_A (|x|))^{2/d}  }    
\] for all $|x|>R$.    Then for any $\ve>0$,  the negative spectrum of  the operator
\[\notag
-\Delta+\ve\sum_{n\in {\Bbb Z}^d}\omega_n \chi(x-n)-V(x)
\]
is almost  surely  infinite.
\end{corollary}

This  corollary is equivalent to Theorem~\ref{sake}.

\section{Proof of  the inequality \eqref{ProbDistrR}}

We now derive the quantitative estimate for the random radius appearing in Theorem~\ref{CRLMS}. The argument uses only the probabilistic estimate obtained in Section 3 and does not require independence of the events associated with different layers.

For each \(n\ge1\), let \(E_n\) denote the event that at least one cube \(Q_{n,\widetilde n}\) intersecting the layer \(\mathcal L_n\) satisfies the low-density condition
$$ \#\{\text{white cells in }Q_{n,\widetilde n}\}\ge (n-1)^d. $$
Thus, with the notation of Section 3,
$$\mathbb P_n=\mathbb P(E_n). $$
By \eqref{Pn<qA},
\[\label{Pa}
\mathbb  P_n \le q^{n^d} A^{dn^d(1+o(1))}.  
\]

We now recall the definition of  the random radius.
Since
$$ \sum_{n=1}^{\infty}\mathbb P_n<\infty $$
whenever \(A^d<q^{-1}\), the Borel--Cantelli lemma implies that, almost surely, only finitely many of the events \(E_n\) occur. 
For a realization \(\omega\) for which only finitely many \(E_n\) occur, let
$$ n_*(\omega)=\max\{n\ge 1:E_n\text{ occurs}\}, $$
with the convention \(n_*(\omega)=1\) if no exceptional layer occurs. We may then choose
$$ R(\omega)=A^{n_*^d(\omega)}.  $$
Thus \(R(\omega)>A\), and by construction all layers  outside  the  ball of radius \(R(\omega)\) contain no exceptional cube.  
In addition,  we choose \(R(\omega)\)  sufficiently large so that all estimates in Sections 3–5 that require \(n\ge n_1(A)\)  are simultaneously valid.
 Then this choice of \(R(\omega)\) is admissible in Theorem~\ref{CRLMS}.

The event
$$ A^{(n-1)^d}<R(\omega)\le A^{n^d} $$
implies that the \(n\)-th layer is exceptional; in other words,
$$ \left\{ A^{(n-1)^d}<R(\omega)\le A^{n^d} \right\} \subseteq E_n.  $$
Therefore,
\[\label{thelatter}
\mathbb P\bigl(  A^{(n-1)^d}<R(\omega)\leq A^{n^d}\bigr)\leq  \mathbb P_n.
\]
Combining  \eqref{Pa} and \eqref{thelatter},  we obtain \eqref{ProbDistrR}.

\end{document}